\documentclass[runningheads,a4paper]{llncs}

\usepackage[T1]{fontenc}
\usepackage{amsmath,amssymb}
\usepackage{graphicx}
\usepackage{booktabs}
\usepackage{array}
\usepackage{colortbl}   
\definecolor{rowtint}{gray}{0.93}
\newcolumntype{R}[1]{>{\raggedright\arraybackslash}p{#1}}
\usepackage{url}
\usepackage{xspace}
\usepackage{tikz}
\usetikzlibrary{arrows.meta,positioning,shapes,fit,backgrounds,calc}
\usepackage{pgfplots}
\pgfplotsset{compat=1.17}
\usepackage[hidelinks]{hyperref}
\usepackage{microtype}
\usepackage{bbding}     
\usepackage{orcidlink}
\makeatletter
\@ifundefined{orcidID}%
  {\newcommand{\orcidID}[1]{\,\orcidlink{#1}}}%
  {\renewcommand{\orcidID}[1]{\,\orcidlink{#1}}}
\makeatother
\newcommand{\sys}{\textsc{RV-Fabric}\xspace}

\begin{document}

\title{A Resilient Runtime-Verification Fabric for Security Monitoring
       of Critical Edge-IoT Infrastructure\thanks{This is a preprint of an article published in the 21st International Conference on Critical Information Infrastructures Security (CRITIS 2026).}}
\titlerunning{A Resilient RV Fabric for Critical Edge-IoT Security Monitoring}
\author{%
Nikolaos Kekatos\,\textsuperscript{(\Envelope)}\inst{1}\orcidID{0009-0000-2358-4270} \and
Marinelio Chintri\inst{2}\orcidID{0009-0009-3923-1790} \and
Panagiotis Katsaros\inst{3}\orcidID{0000-0002-4309-5295} \and
Alexios Lekidis\inst{4}\orcidID{0000-0002-8243-7075} \and
Tom Nianios\inst{1}\orcidID{0009-0008-2238-3555} \and
Ioannis Seitoglou\inst{2}\orcidID{0009-0002-5807-3635} \and
Anastasios Temperekidis\inst{1}\orcidID{0009-0005-5856-2780} \and
Stylianos Basagiannis\inst{2}\orcidID{0000-0002-4513-0541}}
\authorrunning{N. Kekatos et al.}
\institute{%
$^{1}$Clone Systems, Cyprus, \email{\{nkekatos,tnianios,atemperekidis\}@clone-systems.com}\\
$^{2}$International Hellenic University, Greece, \email{\{basagiannis,marchint,ioaseito\}@ihu.gr}\\
$^{3}$Aristotle University of Thessaloniki, Greece, \email{katsaros@csd.auth.gr}\\
$^{4}$University of Thessaly, Greece, \email{alekidis@uth.gr}\\[2pt]
\textsuperscript{(\Envelope)}\,Corresponding author.}
\maketitle

\begin{abstract}
Protecting critical infrastructure increasingly depends on continuously verifying
large IoT fleets against formal security specifications at runtime. Yet the
runtime-verification (RV) pipelines proposed for this task are typically single-host
prototypes whose monitors read a shared log file, with no resilience to the failures
such deployments incur: a crash or overload silently drops events, clock skew corrupts
the ordering metric monitors require, a time-triggered ``node has gone silent'' property
cannot fire when the network itself falls silent, and one slow consumer stalls the
pipeline. Each failure is \emph{silent}: the monitor keeps emitting verdicts over a
corrupted view. We present
\sys, a resilient delivery layer that carries the hierarchy over two brokers (MQTT for
device ingest, a durable stream broker for backend delivery) and re-establishes five
continuity guarantees: durable delivery under crashes, a trusted event order, progress
under total silence, consumer isolation and flow control under bounded overload, each
an invariant conditioned on broker durability.
Above the transport, \sys makes evidence completeness part of runtime-verification
semantics: every verdict carries a status (sound, degraded, incomplete or unavailable)
derived from delivery gaps, retention pressure and liveness, so an incomplete stream
cannot yield an unqualified all-clear. Under controlled fault injection on a containerised
testbed, measured against a fault-free oracle using the real MonPoly engine, the
shared-log baseline misses six of seven injected incidents, reporting each as an
unqualified all-clear, whereas \sys preserves all seven; removing a delivery mechanism
reintroduces silent loss, removing isolation costs only timeliness. Two published
critical-infrastructure datasets, water-SCADA and IoT/IIoT, replay end-to-end.

\end{abstract}

\section{Introduction}
\label{sec:intro}
Security monitoring for the Internet of Things is moving off individual devices onto the
tiered edge-cloud continuum, where evidence from many devices must be collected, ordered,
and reasoned about beyond any single node~\cite{shi2016edge}. This matters because large
IoT fleets are now both critical infrastructure and a favoured attack surface:
internet-facing devices are routinely compromised and conscripted into coordinated
campaigns that a single device, seeing only its own traffic, cannot
detect~\cite{stellios2018iot,cavalli2024iotsec}. Runtime verification (RV), which
continuously checks a running system against explicit formal
specifications~\cite{rvsurvey,bartocci2018introduction}, suits such behaviour: unlike
signature- or learning-based intrusion detection, an RV monitor flags any execution that
violates a stated property and emits an auditable witness. Recent hierarchical designs
push this across a deployment's tiers, with lightweight per-device checks, per-gateway
aggregation, and fleet-wide temporal-logic correlation catching coordinated cross-device
attacks no single node can see~\cite{ares2026,csr2026}.

The verification logic of these designs is mature, but the delivery layer beneath
them is not, and for critical infrastructure that layer is where resilience is
won or lost. Published prototypes run on a single host, with every monitor reading
from one shared append-only log file. That arrangement suits a formal-methods
evaluation but collapses on the edge-cloud continuum:
\emph{(i)} device-to-gateway links are lossy, so records vanish before a monitor
ever sees them; \emph{(ii)} devices and gateways keep independent, skewing clocks,
breaking the monotone-timestamp assumption metric temporal-logic monitors
depend on; \emph{(iii)} time-triggered properties (``no device has reported for
$T$ seconds'') never fire when the network goes silent, because a log
with no new lines produces no events to advance the monitor's clock; and
\emph{(iv)} a single shared log serialises all consumers with no flow
control, so one slow monitor stalls the pipeline. These failures are silent: the monitors keep emitting
verdicts from a corrupted view of the fleet, so a dropped or mis-ordered
security-relevant event may cause an attack to go unflagged, exactly the
continuity failure a critical-infrastructure operator cannot tolerate
(Fig.~\ref{fig:silent}). \emph{In a distributed deployment, no news is not good
news: evidence completeness must be carried inside the verdict.}

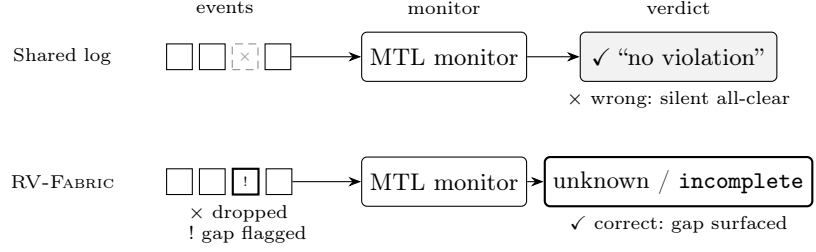
\begin{figure}[t]
\centering
\footnotesize
\begin{tikzpicture}[
  ev/.style={draw, minimum size=3.4mm, inner sep=0pt},
  drop/.style={ev, draw=black!45, densely dashed, text=black!45},
  gap/.style={ev, thick},
  mon/.style={draw, rounded corners=2pt, minimum height=6.5mm, minimum width=15mm, align=center},
  vbad/.style={draw, rounded corners=2pt, minimum height=6.5mm, minimum width=26mm, align=center, fill=black!5},
  vok/.style={draw, rounded corners=2pt, minimum height=6.5mm, minimum width=26mm, align=center, thick},
  lbl/.style={font=\scriptsize, align=left, text=black},
  a/.style={-Stealth}]
\def\mx{3.5}\def\vx{6.6}\def\ry{0.4}\def\fy{-1.25}\def\toff{0.56}
\node[lbl] at (0.6,1.05) {events}; \node[lbl] at (\mx,1.05) {monitor};
\node[lbl] at (\vx,1.05) {verdict};
\node[lbl] (t1) at (-1.55,\ry) {Shared log};
\node[ev] (x1) at (0,\ry) {}; \node[ev,right=0.8mm of x1](x2){};
\node[drop,right=0.8mm of x2](x3){\tiny$\times$}; \node[ev,right=0.8mm of x3](x4){};
\node[mon] (m1) at (\mx,\ry) {MTL monitor};
\node[vbad] (v1) at (\vx,\ry) {\checkmark\ ``no violation''};
\draw[a](x4)--(m1); \draw[a](m1)--(v1);
\node[lbl,align=center] at (\vx,\ry-\toff) {$\times$ wrong: silent all-clear};
\node[lbl] (t2) at (-1.55,\fy) {\sys};
\node[ev] (y1) at (0,\fy) {}; \node[ev,right=0.8mm of y1](y2){};
\node[gap,right=0.8mm of y2](y3){\tiny!}; \node[ev,right=0.8mm of y3](y4){};
\node[mon] (m2) at (\mx,\fy) {MTL monitor};
\node[vok] (v2) at (\vx,\fy) {unknown / \texttt{incomplete}};
\draw[a](y4)--(m2); \draw[a](m2)--(v2);
\node[lbl,align=center] at (\vx,\fy-\toff) {\checkmark\ correct: gap surfaced};
\node[lbl,align=left] at (0.9,\fy-\toff) {$\times$ dropped\\{!} gap flagged};
\end{tikzpicture}
\Description{Two-row schematic comparing a shared-log pipeline with the fabric. In the top row, labelled Shared log, a sequence of event boxes with one dropped event marked by a cross feeds an MTL monitor, which outputs a check-marked verdict reading no violation, annotated as wrong: a silent all-clear. In the bottom row, labelled with the fabric name, the same sequence with the gap flagged by an exclamation mark feeds an MTL monitor, which outputs the verdict unknown / incomplete, annotated as correct: gap surfaced.}%
\caption{The failure mode \sys removes. On a shared log a dropped or mis-ordered event
is invisible to the monitor, which keeps emitting ``no violation'' over a corrupted view
(a \emph{silent} failure); \sys surfaces the breach through event-id gaps, downgrading
the verdict to $(\mathsf{unknown},\mathsf{incomplete})$ (Sect.~\ref{sec:guarantees}).}
\label{fig:silent}
\end{figure}

\paragraph{Contributions.} This paper makes four contributions.
\emph{(1)}~\sys, a resilient delivery layer carrying an unchanged three-layer monitoring
hierarchy~\cite{ares2026,csr2026} across the edge-cloud continuum, re-establishing the
delivery, ordering, liveness, isolation, and flow-control guarantees a shared log cannot
provide (Sect.~\ref{sec:arch}). \emph{(2)}~A threat model for verification-grade
monitoring mapping a tiered network/device/gateway adversary to its mitigation, with the
five continuity properties stated as broker-contingent invariants
(Sects.~\ref{sec:threatmodel}--\ref{sec:guarantees}). \emph{(3)}~An evidence-aware verdict
model, with a conditional-soundness proposition, attaching a completeness status to every
verdict so a degraded stream can no longer produce an unqualified all-clear
(Sect.~\ref{sec:guarantees}). \emph{(4)}~An evaluation quantifying the restored properties
under crash, clock skew, silence, and overload, including verdict preservation against a
fault-free oracle, and replaying two critical-infrastructure datasets through the real
engine (Sects.~\ref{sec:eval},~\ref{sec:casestudy}).

\paragraph{Novelty and positioning.} We claim no new messaging or logging primitive.
MQTT, durable stream brokers, transactional outboxes, deduplication, and heartbeats are
used exactly as they are, and deliberately so: an operator can deploy \sys on components
already present and already certified in their estate, whereas a novel replication
protocol would forfeit precisely that property. What these components do not supply is
the link to the verification layer above them. \sys makes the \emph{completeness of the
evidence stream} an explicit part of runtime-verification semantics: each verdict carries
a continuity status derived from delivery, ordering, retention, and liveness evidence
(Sect.~\ref{sec:guarantees}), so a corrupted or incomplete stream can no longer produce
an unqualified all-clear. A transport layer cannot make this judgement on its own,
because it does not know what the monitor required. The accompanying claim is therefore
one of \emph{necessity} rather than assembly, and we test it: removing any single
delivery mechanism (G1--G3, G5) reintroduces one specific silent failure, from $0$ to
$76$ of $1200$
timestamp-order violations without G2, two silence incidents reverting to false
all-clears without G3, and $1000$ of $2000$ verdicts silently evicted rather than
surfaced without G5 (Table~\ref{tab:ablation}); one level up, on the security verdicts themselves, the
shared-log baseline preserves one of seven injected incidents while the full fabric
preserves all seven, and dropping the tick alone yields two false all-clears that no
completeness status can flag (Sect.~\ref{sec:eval}). That is an argument no individual
component makes.

\section{Background: The Hierarchy and the Shared-Log Prototype}
\label{sec:background}
This section states what \sys inherits: the hierarchical monitoring workload it carries
and the single-host shared-log prototype it replaces, both from prior
work~\cite{ares2026,csr2026}. What \sys adds begins in Sect.~\ref{sec:arch}.

\paragraph{The three-layer hierarchy.}
\sys distributes a hierarchical monitoring workload across three tiers, each seeing a
wider slice of the fleet than the one below. \emph{Layer~1 (L1)}, on or beside each
device, runs stateless per-event sanity checks (payload size, publish rate, timestamp
monotonicity) and emits a compact per-window verdict. \emph{Layer~2 (L2)}, at the
gateway, aggregates the L1 verdicts of the devices it serves. \emph{Layer~3 (L3)}, in
the backend, evaluates fleet-wide metric first-order temporal-logic
(MFOTL~\cite{basin2015mfotl}) properties such as coordinated cross-device attacks and
silent-node failures: first-order quantification with metric bounds (e.g.\ ``two distinct
devices exceed a threshold within $2$\,s'') that relate events at several gateways on one
ordered timeline and are inexpressible at a single tier~\cite{ares2026,csr2026}. The
design is engine-agnostic; at L3 we use the MonPoly~\cite{monpoly} and
RTLola~\cite{rtlola} hybrid of a recent hierarchical design~\cite{ares2026}.

\paragraph{The shared-log prototype and its gaps.}
In that design's prototype all three layers run on one host over a single shared log,
exposing the four failures of Sect.~\ref{sec:intro} (loss, skew, silence, no flow
control). The prototype was reported to drop $\sim$1\% of events at ingest~\cite{ares2026},
but a controlled same-host replay does \emph{not} reproduce this (Sect.~\ref{sec:eval});
the shared log's real exposure is loss under crash or overload, broken per-device ordering
under skew, and silence-liveness. Prior work names an event bus as the
fix~\cite{ares2026,csr2026}; this paper builds that bus and shows it does not suffice on
its own, because the monitor must also know when the bus has failed. Fig.~\ref{fig:arch} shows the resulting architecture.

\section{The \sys Architecture}
\label{sec:arch}
This section derives the transport requirements from each layer's assumptions and presents
the fabric that meets them.

\paragraph{From RV assumptions to transport requirements: the design argument.}
Distributing the hierarchy turns each layer's implicit input assumption into an explicit
transport requirement: a verdict is sound only if the monitor sees a \emph{complete} input
(at-least-once delivery), in \emph{order} (a gateway-anchored timeline), with
\emph{progress} under silence (a periodic tick), read \emph{independently} (per-consumer
cursors), and \emph{without loss under load} (flow control). \sys supplies exactly these
five (Sect.~\ref{sec:guarantees}); each names a way monitoring silently degrades under the
conditions an incident produces.

\begin{figure}[t]
\centering
\resizebox{\columnwidth}{!}{%
\begin{tikzpicture}[
  font=\normalsize,
  src/.style={draw, fill=white, rounded corners=2pt, align=center,
              minimum height=9mm, minimum width=22mm, inner sep=3pt},
  hub/.style={draw, fill=white, thick, rounded corners=2pt, align=center,
              minimum height=12mm, minimum width=17mm, inner sep=3pt},
  proc/.style={draw, fill=white, rounded corners=2pt, align=center,
               minimum height=12mm, minimum width=19mm, inner sep=3pt},
  arrow/.style={-Stealth, thick},
  flow/.style={font=\footnotesize\itshape, align=center},
  tier/.style={font=\footnotesize\bfseries},
  band/.style={rounded corners=3pt, inner sep=3.5mm}]
\node[src] (mc)               {Device\\(L1)};
\node[src, below=3mm of mc]   (sens) {Device\\(L1)};
\node[src, below=3mm of sens] (leg)  {Device\\(L1)};
\node[hub, right=12mm of sens] (mqtt) {MQTT\\(QoS\,1)};
\node[proc, right=10mm of mqtt] (gw) {Gateway\\L1/L2\\+ sidecar};
\node[hub, right=16mm of gw, minimum width=22mm] (nats) {NATS\\JetStream};
\node[proc, right=12mm of nats] (l3) {L3\\(MFOTL,\\cursors)};
\node[right=7mm of l3, align=center, font=\footnotesize] (out) {incidents\\+\,witnesses};
\draw[arrow] (mc.east)  -| (mqtt.north);
\draw[arrow] (sens.east) -- (mqtt.west);
\draw[arrow] (leg.east)  -| (mqtt.south);
\draw[arrow] (mqtt) -- node[flow, below] {ingest} (gw);
\draw[arrow] (gw)   -- node[flow, below] {verdicts} (nats);
\draw[arrow] (nats) -- node[flow, below] {read} (l3);
\draw[arrow] (l3)   -- (out);
\node[font=\scriptsize\itshape, gray!55!black, below=13mm of nats.south]
   (tick) {clock-tick 1\,Hz};
\draw[-Stealth, dashed, thin, gray!55!black] (tick) -- (nats.south);
\draw[-Stealth, dashed, thin, gray!55!black] (nats.south west) to[out=-140, in=-40]
   node[font=\scriptsize\itshape, gray!55!black, pos=0.5, below=1mm] {flow control} (gw.south east);
\begin{scope}[on background layer]
  \node[band, fill=black!4, fit=(mc)(sens)(leg)] (bdev) {};
  \node[band, fill=black!9, fit=(mqtt)(gw)]      (bgw)  {};
  \node[band, fill=black!4, fit=(nats)(l3)(out)] (bbe)  {};
\end{scope}
\node[tier, above=1mm of bdev.north] {Device tier};
\node[tier, above=1mm of bgw.north]  {Gateway (L1/L2)};
\node[tier, above=1mm of bbe.north]  {Backend (L3)};
\end{tikzpicture}}
\Description{Block diagram of the three-tier architecture read left to right. The Device tier holds three L1 device boxes publishing into an MQTT QoS 1 broker in the Gateway tier, which forwards ingest to a Gateway L1/L2 node with a timestamping sidecar. The gateway publishes verdict streams to a NATS JetStream broker in the Backend tier, read by an L3 monitor running MFOTL with durable cursors, which emits incidents and witnesses. A dashed 1 Hz clock-tick input feeds JetStream, and a dashed flow-control arrow runs back from JetStream to the gateway.}%
\caption{\sys: the three-layer RV hierarchy decoupled over a two-broker fabric.
Devices publish at QoS\,1 to a gateway-partitioned MQTT ingest (logically one broker per
gateway; the prototype uses per-gateway topic prefixes on one shared broker, equivalent
for routing but without broker-level fault isolation); the gateway runs L1/L2 and a
timestamping sidecar and re-publishes verdict streams to a backend NATS JetStream cluster,
from which the L3 fleet monitors consume via independent durable cursors. A 1\,Hz clock-tick subject drives the time-triggered silent-node
monitors, and flow control defers loss under transient overload.}
\label{fig:arch}
\end{figure}
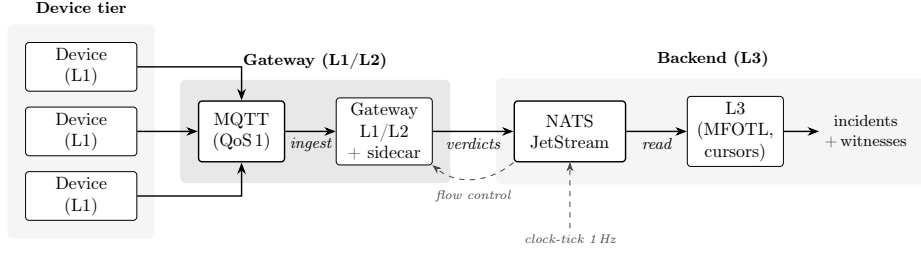

\paragraph{MQTT for device-to-gateway ingest.}
The device-to-gateway hop is constrained and lossy, so \sys uses MQTT (Mosquitto,
QoS\,1)~\cite{mqtt,mosquitto}, which is ubiquitous on IoT and gives at-least-once delivery
with explicit PUBACKs; L1 is idempotent per $(\text{device},\text{window})$, so duplicates
are harmless.

\paragraph{NATS JetStream for gateway-to-backend delivery.}
The gateway-to-backend hop must fan a gateway's verdict stream to several
independent L3 monitors without a slow one blocking the others. \sys uses NATS
JetStream~\cite{nats}: per-consumer durable cursors let each L3 monitor read at
its own pace, and a dedicated 1\,Hz tick publisher writes to a JetStream subject
as the heartbeat that time-triggered monitors need. Verdicts are keyed by
per-device subject (\texttt{gw.$\langle$id$\rangle$.verdict}), so each device's stream
stays FIFO-ordered while a wildcard subscription (\texttt{gw.*.verdict}) fans a whole
gateway's devices into one monitor.

\paragraph{Timestamping sidecar.}
A sidecar co-located with each gateway stamps each event with a trusted, non-decreasing
gateway-\emph{ingest} time and orders by it, restoring at the first trust boundary the
ordering invariant L3 assumes. Device-reported time is retained only for
spoof detection, never for ordering; original device event time is not reconstructed.

\paragraph{Durable relay and deduplication.}
On receipt the gateway assigns each verdict a persistent event id
($\langle$gw, device, seq$\rangle$) and appends it to a local \texttt{fsync}ed write-ahead
outbox \emph{before} forwarding it to JetStream; once JetStream acknowledges the durable
publish, an \texttt{fsync}ed confirmation cursor advances, so a restart replays only the
un-acked suffix, bounding recovery cost. QoS\,1 with persistent sessions protects the
device-to-broker hop, the outbox the gateway-to-JetStream relay. Backend consumers
deduplicate by event id, so replays never double-count, giving exactly-once monitoring
\emph{effect} over an at-least-once transport without distributed
transactions~\cite{helland2007life}. A narrow residual window remains between delivery of
the brokered verdict to the relay and its local WAL \texttt{fsync}, where the verdict is
still only in memory; Sect.~\ref{sec:eval} injects a crash at each relay stage.

\paragraph{Properties the fabric carries.}
End-to-end, \sys carries a hierarchical security-property
catalogue~\cite{ares2026,csr2026}: single-node checks at L1/L2 and, at L3,
cross-device correlation (P3.1--P3.4) plus a time-triggered silent-node property
(P3.5). Briefly, P3.1 flags a multi-vector advanced persistent threat (APT) on one device, P3.2
the same attack from $\geq3$ devices, P3.3 an escalation over time, P3.4 a persistent
overflow campaign, and P3.5 a device gone silent, each sound only because a specific
transport guarantee holds (Sect.~\ref{sec:guarantees}). The
multi-gateway view also lets \sys express a property the single-gateway prototype
cannot observe: P3.6 flags an overflow campaign spanning $\geq2$ distinct
gateways,
\begin{equation}
P_{3.6}\;\equiv\;\big(\mathit{cnt}\leftarrow\mathsf{CNT}\,g;\;
\text{\small ONCE}_{[0,30]}\,\mathit{overflow}(g)\big)\;\wedge\;\mathit{cnt}\ge 2,
\label{eq:p36}
\end{equation}
i.e.\ at least two distinct gateways report an overflow within a 30\,s window.
A second fabric-enabled property, P3.7, fires when an entire gateway falls silent
while the fabric's tick keeps pulsing, distinguishing gateway-level stream loss
from the silence of an individual device:
$P_{3.7}(g)\equiv(\text{\small ONCE}\,\mathit{verdict}(g))\wedge
\neg\,\text{\small ONCE}_{[0,T_{\mathrm{gw}}]}\,\mathit{verdict}(g)$, evaluated at each tick.
P3.7 is a \emph{fabric-enabled availability} property rather than an intrinsically
cross-gateway one: it requires only a clock independent of the monitored stream (unlike
P3.6, which genuinely needs $\geq2$ gateway identities). That clock is a
\emph{backend-resident} tick publisher advancing monitoring time independently of every
gateway stream, so silence-liveness is by construction independent of any gateway remaining
up, where a per-gateway tick would freeze.

\section{Threat Model and Trust Assumptions}
\label{sec:threatmodel}
\sys is security-monitoring infrastructure, so its guarantees must be judged against an
explicit adversary, not benign-failure assumptions alone. This section states what \sys
defends, what it trusts, and which guarantees hold against which adversary, the threat
model the continuity mechanisms of Sect.~\ref{sec:guarantees} are designed against.

\paragraph{Asset and adversary goal.}
The asset \sys protects is the \emph{monitoring verdict stream}: every
security-relevant event durably accepted into the protected path must reach the
correlator, in the trusted gateway-ingest order, and the \emph{absence} of events must
itself be observable. The IoT payload is out of scope; \sys defends the evidence about
it, not the data. The adversary's goal is to make a
genuine attack go \emph{unflagged}, by corrupting the monitor's view
(drop/reorder/replay), \emph{blinding} it (silencing a device or gateway), or
\emph{exhausting} it (overload that sheds events), the silent-degradation modes of
Sect.~\ref{sec:intro} recast as adversarial objectives.

\paragraph{Trusted computing base.}
\sys trusts the cloud backend and L3 monitors, the two brokers under their configured
durability, and each gateway's timestamping sidecar and durable outbox, a
deliberately small TCB on the operator's side. Everything device-side of a gateway,
the devices and the network links, is untrusted; link cryptography (broker TLS and
authentication) is assumed, so the network adversary can manipulate traffic but cannot
forge a broker-authenticated identity.

\begin{table}[t]
\centering
\caption{Threat model: for each adversary capability, the \sys mitigation, the
continuity mechanism (G1--G5, Sect.~\ref{sec:guarantees}) or monitored property (P3.x)
it uses, and the residual risk. Entries marked $\dagger$ lie outside the trust boundary
(assumed, not enforced).}
\label{tab:threat}
\footnotesize
\setlength{\tabcolsep}{4pt}
\begin{tabular}{@{}R{2.6cm}R{4.0cm}R{1.0cm}R{3.0cm}@{}}
\toprule
\textbf{Adversary capability} & \textbf{\sys mitigation} & \textbf{Via (G/P)} & \textbf{Residual risk} \\
\midrule
\rowcolor{rowtint}
\multicolumn{4}{@{}l}{\emph{Network adversary (untrusted links)}}\\[1pt]
drop / delay events & at-least-once (QoS\,1 \& outbox); flow control absorbs transient backlog & G1, G5 & observable retention exhaustion, then loss \\
reorder events & monotone gateway timestamps re-anchor the timeline & G2 & cross-gateway skew bounded by offset $\epsilon$ (Sect.~\ref{sec:eval}) \\
replay events & event-id dedup (idempotent monitoring effect) & n/a & ID collision / sequence-state loss; ID forgery$^\dagger$ \\
\cmidrule{1-4}
\rowcolor{rowtint}
\multicolumn{4}{@{}l}{\emph{Device adversary}}\\[1pt]
emit malicious traffic & L1/L2 checks; fleet correlation & P3.1--P3.4 & novel behaviour outside the specifications \\
spoof its clock & sidecar ignores device time for ordering & G2 & device time trusted only for the spoof check \\
go silent (jam / DoS) & backend tick drives silent-node under total silence & G3, P3.5 & cause (fault vs.\ attack) not identified \\
\cmidrule{1-4}
\rowcolor{rowtint}
\multicolumn{4}{@{}l}{\emph{Gateway adversary}}\\[1pt]
taken dark & gateway-silence via the backend-resident tick & P3.7 & that fleet's in-flight evidence is lost \\
compromised$^\dagger$ & \emph{out of scope}: sidecar integrity trusted & n/a & full loss of that fleet's monitoring soundness \\
\bottomrule
\end{tabular}
\end{table}

\paragraph{Adversary treatment and scope.}
Table~\ref{tab:threat} maps each capability to its mitigation. The key property is that an
adversary who \emph{silences} a link cannot thereby hide the silence: P3.5/P3.7 turn
suppression into a positive detection, though evidence may be delayed during a prolonged
partition. The residual is behaviour that is malicious yet specification-conforming, the
general limit of any specification-based monitor. A gateway \emph{taken dark} is handled
(P3.7 plus outbox replay on recovery); an \emph{actively compromised} one is out of scope,
its sidecar and outbox lying in the TCB, so \sys bounds the blast radius by
keeping fleets disjoint (Sect.~\ref{sec:impl}).

\section{Continuity Mechanisms and Their Invariants}
\label{sec:guarantees}
This section presents the continuity mechanisms \sys establishes over the two brokers,
states each as a broker-contingent invariant tied to a specific shared-log failure, and
defines the verdict algebra coupling them to the output.

\sys provides five continuity mechanisms, which we label as \emph{guarantees}
\textbf{G1}~Delivery, \textbf{G2}~Order, \textbf{G3}~Tick, \textbf{G4}~Isolation, and
\textbf{G5}~Flow (the letter G stands for \emph{guarantee}).
We state each as an \emph{invariant} of the fabric, contingent on the durability
configuration of the components that supply it (Sect.~\ref{sec:impl}): their value is
the explicit correspondence, RV assumption $\to$ transport property $\to$ shared-log
failure closed.

Two carry the most weight. The \emph{clock-tick} is the subtlest: a time-triggered ``no
device has reported for $T_{\mathrm{dev}}$ seconds'' property cannot fire from a shared
log, since a log with no new lines produces no events to advance the monitor's clock, so a
1\,Hz tick supplies that liveness under total silence. \emph{Flow control} only
\emph{defers} loss by letting events queue: sustained overload beyond the configured
retention capacity, evaluated in Q7, makes exhaustion observable but cannot preserve G1
beyond capacity, so delivery is conditional on capacity.

\paragraph{Guarantees as invariants.}
Write $E_{\mathrm{brk}}$ for messages the ingest broker has durably accepted,
$E_{\mathrm{gw}}$ for those the gateway receives, $E_{\mathrm{wal}}$ for verdicts
\texttt{fsync}ed to the gateway WAL, and $\hat{E}$ for what the backend L3 monitor
observes; each verdict $v$ carries a gateway timestamp $\tau(v)$. The five guarantees
are fabric invariants:
\emph{(G1)}~at-least-once on two protected segments,
$E_{\mathrm{brk}}\subseteq E_{\mathrm{gw}}$ (persistent MQTT hop, given durable QoS\,1
acceptance under a persistent session) and $E_{\mathrm{wal}}\subseteq\hat{E}$
(gateway--backend relay), so nothing durably accepted is lost, while drops \emph{before}
durable acceptance (device crash, device-side queue, pre-WAL window) lie outside G1;
\emph{(G2)}~monotonicity, if $v_1$ is delivered before $v_2$ then
$\tau(v_1)\le\tau(v_2)$ for same-device verdicts;
\emph{(G3)}~liveness, a 1\,Hz tick drives the backend even when $\hat{E}=\emptyset$;
\emph{(G4)}~isolation, each consumer's cursor never gates the others; and
\emph{(G5)}~flow control, awaited publishes and retention queue transient overload
rather than dropping it. G1--G2 protect and re-order the monitoring view; G3--G5
enable resilient distributed operation.

\paragraph{Evidence-aware verdicts.}
The mechanisms above bound \emph{when} the input stream is trustworthy; a monitor that
keeps emitting a bare ``no violation'' after the stream is corrupted reproduces the silent
failure we set out to remove. We therefore make evidence completeness explicit in the
verdict semantics. Following three-valued RV semantics~\cite{bauer2011rvltl}, which admit
an \textsf{inconclusive} verdict when the observed trace is insufficient, a verdict is a
pair $V=(s,c)$ of a \emph{security} and an \emph{evidence-completeness} component,
\[
\begin{aligned}
  s &\in \{\mathsf{violation},\ \mathsf{no\_violation},\ \mathsf{unknown}\},\\
  c &\in \{\mathsf{sound} \succ \mathsf{degraded} \succ \mathsf{incomplete} \succ \mathsf{unavailable}\},
\end{aligned}
\]
where $c$ is derived from continuity evidence the fabric already tracks:
$c{=}\mathsf{sound}$ when every event on the protected segments is accounted for
(G1, no gap); $\mathsf{degraded}$ past a soft retention high-water mark (G5) before
any loss; $\mathsf{incomplete}$ when a sequence gap, retention overflow, or a
monotonicity violation on the consumed stream (G2) means events are provably missing or
misordered; and $\mathsf{unavailable}$ when the tick (G3) shows the
monitor live but no evidence is arriving. Crucially, $\mathsf{unavailable}$ is an
evidence-\emph{quality} tag, not a verdict: for a property whose violation \emph{is}
the absence of events (the silent-node P3.5 and gateway-silence P3.7), the tick
reliably witnesses that absence, so it yields $(\mathsf{violation},\mathsf{sound})$;
$\mathsf{unavailable}$ applies only to properties needing positive event evidence.

The two components interact by one rule that keeps security and completeness as
\emph{separate} axes rather than as one ad hoc lattice. For the monotone witness
properties formalised in Prop.~\ref{prop:soundness}, a \emph{positive} finding backed
by an intact local witness is self-evidencing (the witness \emph{is} the evidence) and
remains valid under \emph{unrelated} missing evidence, so $(\mathsf{violation},c)$
stands for any $c$; a \emph{negative} finding is only as good as its input, so a
$\mathsf{no\_violation}$ over non-\textsf{sound} evidence is downgraded,
$(\mathsf{no\_violation}, c)\mapsto(\mathsf{unknown}, c)$ for $c\prec\mathsf{sound}$,
stopping a corrupted stream from yielding an unqualified all-clear.
At L3 a property spanning gateways $g_1,\dots,g_n$ carries the meet of their evidence,
$c_{\mathrm{L3}}=\min_{\preceq}\{c_{g_1},\dots,c_{g_n}\}$, so one incomplete gateway makes
a fleet-wide all-clear $\mathsf{unknown}$. We implement all four tags over the real
JetStream (event-id continuity for gaps, a consumed-stream monotonicity check for order
violations, the retention watermark for $\mathsf{degraded}$, tick-versus-evidence for
$\mathsf{unavailable}$) and confirm each on a canonical scenario: a normal stream stays
\textsf{sound}, an $85\%$ retention high-water mark is \textsf{degraded}, an event-id gap
is \textsf{incomplete} (downgrading \textsf{no\_violation} to \textsf{unknown}), and a
required stream silent under a live tick is \textsf{unavailable}.

\paragraph{Soundness of the verdict algebra.}
The downgrade rule is not merely a convention: it makes the algebra \emph{sound} for a
well-defined and practically dominant class of properties. Call $\varphi$ a
\emph{monotone witness property} if its satisfaction over a window is certified by a
finite set of events and preserved under supersets, i.e.\ for traces
$\hat\sigma\subseteq\sigma$, $\hat\sigma\models\varphi\Rightarrow\sigma\models\varphi$.
The counting and multi-vector L3 properties P3.1--P3.4 and P3.6 (existential
\textsc{once}-formulas over device or gateway events) are of this class. The
absence-based P3.5/P3.7 are dual: G3 gives their liveness while G1/G5 rule out
loss-induced false alarms; we treat the monotone class here. Let $\sigma$ be the durably accepted event stream for the monitored identities
over the window (events that reached the protected path; device-side and pre-WAL losses
lie outside G1 and hence outside this result), and $\hat\sigma=\mathrm{dedup}(\hat E)$
the stream L3 observes after event-id deduplication.

\begin{proposition}[Conditional soundness of verdicts]
\label{prop:soundness}
Assume the order guarantee G2 (so $\hat\sigma$ is order-faithful) and, for the
cross-gateway property P3.6, a bounded cross-gateway skew $\epsilon<W$ for the P3.6 correlation window $W$ ($30$\,s); the hypothesis
$c=\mathsf{sound}$ used below itself presupposes G1, G3, and G5 over the window. For any
monotone witness property $\varphi$ and emitted verdict $V=(s,c)$:
\emph{(i)~Positive soundness.} If $s=\mathsf{violation}$ then $\sigma\models\varphi$:
with the deduplication state intact, no loss or overload manufactures a false violation.
\emph{(ii)~Negative soundness under completeness.} If $s=\mathsf{no\_violation}$ and
$c=\mathsf{sound}$ then $\sigma\not\models\varphi$; a sound all-clear agrees with the
fault-free oracle over the accepted stream.
\emph{(iii)~No silent false all-clear.} If $\sigma\models\varphi$ but the monitor would
emit $\mathsf{no\_violation}$ (a missed violation), and the loss is observable in the
evidence stream, as an event-id discontinuity (G1/G5) or an order violation on the
consumed stream (G2), then $c\prec\mathsf{sound}$, so the algebra emits
$(\mathsf{unknown},c)$ rather than an unqualified all-clear.
\end{proposition}

\begin{proof}[sketch]
\emph{(i)}~By intact event-id deduplication (exactly-once \emph{effect}),
$\hat\sigma\subseteq\sigma$: redelivery cannot inflate a count and loss only removes
events, while G2 keeps the \textsc{once}-windows aligned (for P3.6, $\epsilon<W$ aligns the
cross-gateway window). A satisfying witness $\Omega\subseteq\hat\sigma$ gives
$\Omega\subseteq\sigma$, hence $\sigma\models\varphi$.
\emph{(ii)}~$c=\mathsf{sound}$ certifies, over the window, no event-id gap on the
protected segments (G1: no delivery loss), backlog below the high-water mark (G5: no
eviction), and no order violation (G2), so $\hat\sigma$ equals $\sigma$ restricted to
the window's relevant events and
$\hat\sigma\not\models\varphi\Rightarrow\sigma\not\models\varphi$.
\emph{(iii)}~Contrapositive of (ii): a missed violation means the accepted witness was
lost or misordered in $\hat\sigma$; broken event-id continuity, or an order violation
on the consumed stream, each set $c=\mathsf{incomplete}\prec\mathsf{sound}$, triggering
$(\mathsf{no\_violation},c)\mapsto(\mathsf{unknown},c)$.
\qed
\end{proof}

The dual absence-based P3.5/P3.7 obtain the analogous guarantee through the tick, which
times the silence G1/G5 show to be real. Clause~(iii) is \emph{conditional on the loss being
observable}, and that condition is necessary rather than incidental: event-id continuity
makes delivery loss observable, a monotonicity check on the consumed stream makes
reordering observable, and the tick makes silence observable. A transport carrying none of
them (the shared log) cannot lower $c$, so every missed violation there is a silent false
all-clear; equally, without the tick a genuine silence is unobservable, the one loss class
the algebra cannot flag. The oracle experiment (Q8, Sect.~\ref{sec:eval}) validates all
three clauses empirically.

\section{Implementation}
\label{sec:impl}
This section details how the architecture of Sect.~\ref{sec:arch} is realised: the two
brokers, the naming scheme routing events, and the mapping onto the continuum.

\paragraph{Brokers and subjects.}
Each gateway ingests device telemetry over MQTT (Mosquitto~\cite{mosquitto}) at QoS\,1
on topics \texttt{dev/$\langle$gw$\rangle$/$\langle$id$\rangle$/evt}; the two-gateway
testbed shares one broker with per-gateway topic prefixes (a dedicated broker per gateway
is a drop-in alternative), so a gateway consumes only the fleet homed to it. Its L1/L2
consumer re-publishes per-window verdicts to a backend NATS JetStream~\cite{nats} stream,
routed by \emph{subject}: one subject per device, plus a fleet-wide \texttt{fleet.tick}
subject carrying the 1\,Hz liveness heartbeat.

\paragraph{Monitors.}
\sys is deliberately monitor-agnostic: it reuses the L1/L2/L3 monitors of the
hierarchical design under study~\cite{ares2026,csr2026} unmodified. L1/L2 are
lightweight; L3 is the hybrid of MonPoly~\cite{monpoly} for metric first-order
correlation and RTLola~\cite{rtlola} for the time-triggered silent-node property, the
latter driven by the fabric's 1\,Hz tick subject. Because only the transport beneath them
changes, the specification logic is unchanged, anchored to the gateway-ingest
timeline, so an operator adopts \sys without rewriting specifications.

\paragraph{Deployment topology.}
L1 runs at the edge; each gateway node runs an L1/L2 consumer and a sidecar over a
disjoint fleet, so gateways fail independently; a cloud
backend runs JetStream and the L3 monitors. The testbed instantiates two gateway nodes of
four emulated devices each, and scales out by adding nodes.

\paragraph{Reproducibility.}
\emph{Host.} The containerised testbed runs on a single laptop, an Apple M4 Pro
MacBook Pro (12 cores: 8 performance + 4 efficiency; 24\,GB RAM; macOS~15.3), with
Docker~28.1.1. The two-host clock-skew experiment
(Sect.~\ref{sec:eval}) instead uses two separate cloud VMs
(1\,vCPU, Ubuntu~24.04, independent kernels, NTP disabled on one) to obtain real
cross-host skew.
\emph{Software.} The stack comprises Mosquitto~2 (MQTT), NATS~2 JetStream, the
Python gateway/sidecar/relay, and an L3 image that builds MonPoly (OCaml) and
RTLola (Rust) from source; emulated devices replay the workload traces. Everything is
wired with Docker Compose and a single command brings up the fabric plus a
coordinated-attack overlay. Specifications, compose
files, and monitor wiring are available as a public
artefact.\footnote{\url{https://github.com/nikos-kekatos/resilient-rv-fabric-edge-iot}}

\section{Evaluation on the Controlled Testbed}
\label{sec:eval}
This section evaluates whether \sys preserves its five continuity mechanisms under
fault injection and at fleet scale, and quantifies the latency and recovery cost of
doing so. The mechanisms are the means; what is ultimately under test is evidence completeness, so
the sequence closes with Q8, asking whether the security verdicts themselves survive and
whether those that do not are at least reported as incomplete rather than as an
all-clear.

\paragraph{Testbed and methodology.}
We deploy \sys on a multi-gateway testbed: $G$ gateways each serving a share of $N$
emulated devices that replay a documented attack and benign
workload~\cite{ares2026,csr2026}, brokers first colocated and then separated by an emulated
wide-area link, all containerised on the host detailed under Reproducibility above.
Detection figures are means over five runs; per-event latency (Q2) figures are means and
p95s over
a paced sample run per link setting, with WAN delay injected via Linux \texttt{netem}. We ask eight questions, each naming the guarantee (G1--G5) or monitored property (P3.x)
it targets: (Q1)~\emph{correctness}, does the fabric close the documented gaps (G1--G3);
(Q2)~\emph{overhead}, the per-event latency added; (Q3)~\emph{throughput}, the sustained
message rate (G4--G5); (Q4)~\emph{detection} of cross-tier, cross-gateway and silent-node
attacks (P3.1--P3.7); (Q5)~\emph{ablation}, the failure removing each mechanism
reintroduces (G1--G5); (Q6)~\emph{crash tolerance} at each relay stage (G1);
(Q7)~\emph{overload horizon} before continuity is lost (G5); and (Q8)~\emph{verdict
preservation} against a fault-free oracle (G1--G5).

\paragraph{Where the fabric differs (Q1).}
We compare \sys against the single-host shared log (the same hierarchy \emph{without} the
broker fabric), replaying one fixed $7000$-event trace through both at matched load. For
steady-state ingest the two are indistinguishable: \emph{both} deliver every event to L1
($0$ loss, at $500$ and $3000$\,msg/s, five trials each).
The $\sim$1\% ingest loss reported for the earlier prototype~\cite{ares2026} did not
reproduce, so it is condition-specific and not where the fabric differs. It differs on
two properties the evaluated shared-log design lacks. \emph{Ordering}: on the
same $1200$-verdict stream, device-time ordering leaves $76$ out of order under the
time-spoof profile, the gateway sidecar $0$. \emph{Silence-liveness}: under total silence the shared log yields no events so P3.5 never
fires, whereas the 1\,Hz tick fires it within one tick of the $5$\,s deadline.

\paragraph{Per-event latency under a wide-area link (Q2).}
We place the fabric clients behind an emulated wide-area link (Linux \texttt{netem},
one-way delay $\delta$ per hop) and measure publish-to-consume latency per hop (mean/p95
in ms, paced 400-sample run). Colocated, both hops are sub-millisecond in the mean ($0.5/1.0$ MQTT,
$0.9/1.8$ JetStream). They scale differently: the MQTT hop grows as ${\approx}3\delta$
because QoS\,1's acknowledged handshake crosses the link several times (mean $17/32/60$\,ms
at $\delta{=}5/10/20$\,ms), the JetStream hop as ${\approx}\delta$ ($7/13/23$\,ms). Even at
$\delta{=}20$\,ms, p95 stays under $75$\,ms, over an order of magnitude below the
second-scale L3 detection latency, so transport was not the bottleneck and QoS\,1's
steeper WAN cost buys at-least-once delivery on the lossy edge hop.

\paragraph{Throughput and fleet scale (Q3).}
A saturation sweep (one durable consumer, rising producer count) climbs from
$\sim$10.7K to a $\sim$14K\,msg/s plateau at zero loss; past the drain rate, offered
load queues rather than drops (p95 rising to $\sim$1\,s: flow control defers loss).
Emulated fleets of $250/500/1000$ devices ($\approx$2\,msg/s each) all sustain
\emph{zero loss}; at $1000$ devices the $\approx$1930\,msg/s aggregate runs at
$\sim$60\,ms p95. This validates the message-handling path at fleet scale, not the
resource limits of $1000$ physical devices; read capacity scales with more L3 consumers.

\paragraph{Detection (Q4).}
An eight-device, two-gateway workload (mixed overflow/APT/spam/time-spoof/clean) over
five $150$\,s runs fired all six properties P3.1--P3.6 on every run ($41.6\pm2.1$
incidents), including the fabric-dependent ones: cross-device \textbf{P3.2} and
cross-gateway \textbf{P3.6} need the backend fleet view and per-consumer cursors (P3.6 is
unobservable single-gateway), and time-triggered \textbf{P3.5} surfaces only because the
tick drives polling under silence. Taking one gateway dark mid-run, \textbf{P3.7} flags
gateway-stream loss within $T_{\mathrm{gw}}{=}10$\,s, distinct from a single device
going quiet, though not the cause.

\paragraph{Clock-skew resilience of P3.6 (Q4, two independent hosts).}
We validate P3.6's skew resilience on \emph{two independent machines} (cloud VMs, separate
kernels, real network, NTP disabled on gw2 so its clock diverges genuinely). A per-gateway
clock-beacon measures the realised offset $\epsilon$ against a backend reference, with a
zero-skew control at $\epsilon{=}0.001$\,s. Sweeping gw2's injected offset over
$0/5/10/20/35$\,s, the measured $\epsilon$ tracks it to within network jitter
($0.001/4.75/9.71/19.64/33.76$\,s), and P3.6 (``$\geq2$ gateways within $W{=}30$\,s'')
stays sound while $\epsilon<W$ and flips once $\epsilon$ exceeds the window, as expected
for a windowed property. A real cross-region WAN and correlated multi-gateway failure
remain future work (Sect.~\ref{sec:limits}).

\paragraph{Failure mode on mechanism removal (Q5).}
Each guarantee is characterised by the failure its removal reintroduces
(Table~\ref{tab:ablation}), rather than by a single scalar. Consumer isolation (G4) we
measure directly: a fast L3 consumer drains $2000$ verdicts in $0.043\pm0.001$\,s (five
runs) while a slow peer ($20$\,ms/verdict) on its own durable cursor has processed only
$3$, a $\sim$940$\times$ isolation against the $40$\,s a single shared cursor would
impose by head-of-line blocking. The other four are quantified elsewhere (Table~\ref{tab:ablation}). Q8 then carries each
of these transport-level failures one level up, to the security incidents each disabled
mechanism causes to be missed or silently mis-reported.

\begin{table}[t]
\centering
\caption{Per-mechanism ablation: the failure each continuity mechanism's removal
reintroduces, and the measured effect (with vs.\ without the mechanism). Numbers are
from the experiments cross-referenced in the last column.}
\label{tab:ablation}
\footnotesize
\setlength{\tabcolsep}{4pt}
\begin{tabular}{@{}R{2.0cm}R{4.1cm}R{4.6cm}@{}}
\toprule
\textbf{Mechanism} & \textbf{Failure on removal} & \textbf{Measured effect (with $\to$ without)} \\
\midrule
\textbf{G1}~Delivery & crash / overload loses un-acked evidence & $1$ in-flight event lost if the crash precedes WAL \texttt{fsync}, $0$ after; without the outbox the whole un-acked suffix is at risk (Q6) \\
\textbf{G2}~Order & metric monitors see a reordered stream & $0\to76$ of $1200$ timestamp-order violations under time-spoof (Q1) \\
\textbf{G3}~Tick & silent-node never fires under total silence & P3.5 fires within one tick $\to$ $0$ pulses / no firing when the tick is absent (Q1) \\
\textbf{G4}~Isolation & one slow consumer stalls all (head-of-line) & fast drain $0.043$\,s $\to$ $40$\,s shared cursor ($\sim$940$\times$) (Q5) \\
\textbf{G5}~Flow & transient overload dropped, not queued & discard-new surfaces every drop $\to$ discard-old silently loses $1000/2000$ (Q7) \\
\bottomrule
\end{tabular}
\end{table}

\paragraph{Crash recovery of the durable outbox (Q6).}
We inject a hard crash (\texttt{SIGKILL}-equivalent) into the gateway relay at each of
its three stages while it relays $50$ verdicts (crash on \#25), then restart and read the
whole stream through backend deduplication. Crashing \emph{before} the local WAL
\texttt{fsync} (S1), while the verdict is still only in memory, loses that one in-flight
verdict, the residual window of Sect.~\ref{sec:arch}. Crashing \emph{after} it loses
nothing: persisted-but-unacknowledged (S2) and acknowledged-but-cursor-not-advanced (S3)
both lose $0$, S3 re-publishing one duplicate that deduplication collapses to $0$
double-counted incidents (exactly-once monitoring \emph{effect}).
Recovery is linear in the backlog: a lost cursor forcing a full replay takes
$8.8$/$95$/$474$\,ms for $100$/$1000$/$5000$ entries ($\approx$0.095\,ms/entry).

\paragraph{Overload buffer horizon (Q7).}
Flow control defers loss only while the backend buffer holds; we measure what happens
beyond it. Filling a bounded JetStream stream ($C{=}1000$, pushed to $2C$ with no
consumer) under the two discard policies shows the operational choice sharply: default
\emph{discard-old} silently evicts $1000$ of $2000$ verdicts with $0$ publish errors,
whereas \emph{discard-new} raises $1000$ explicit publish errors, making the same loss
\emph{observable}. Crossing the soft high-water mark marks the stream \textsf{degraded};
the first rejected or evicted publication marks it \textsf{incomplete}, consistent with
the algebra of Sect.~\ref{sec:guarantees} (retaining or retrying the rejected event is
future work, Sect.~\ref{sec:limits}). The buffer absorbs an outage of
$T_{\text{buffer}}=C/R_{\text{in}}$: at $\approx$1930\,msg/s, a million-message buffer
covers $\approx$8.6\,min.

\paragraph{Verdict preservation against a fault-free oracle (Q8).}
The ablation (Q5) shows that removing each mechanism reintroduces a transport-level
failure; this experiment measures the consequence one level up, on the
\emph{security verdicts} themselves. We fix a seeded workload that triggers all seven
properties P3.1--P3.7 and take the incident set the \emph{unmodified} MonPoly engine
(with tick-driven reference detectors for the liveness properties P3.5/P3.7) produces
from the fault-free stream as the ground-truth \emph{oracle} $I^\star$
($|I^\star|{=}7$). We then inject one seeded fault campaign, replayed as an identical alert sequence under the shared log,
the full fabric, and the fabric with each guarantee removed, its faults realised by the
measured behaviour of Q6--Q7 and each bound to one incident: a pre-\texttt{fsync} crash
loses \texttt{e1}'s fifth overflow (P3.4's count $5\!\to\!4$, G1); \texttt{c1}'s overflow
is transposed after its time-anomaly, breaking the P3.1 window (G2); \emph{discard-old}
evicts backlog, dropping the P3.2 botnet count (G5); node \texttt{f1} and a third gateway
\texttt{gw3} go silent at fixed offsets for P3.5/P3.7 (G3); and the P3.6 campaign spans
two gateways. The P3.3 escalation on \texttt{d1} is a device-local \emph{control} no fault
touches, and is the single incident the shared log preserves. Since each incident is hit
by at most one fault, the one-miss-per-mechanism structure is a property of the seeded
campaign, not per-configuration tailoring; ids, timestamps, and bindings are pinned in the
artefact (\texttt{exp\_oracle.py}). The only variable is which alerts, and in what order,
reach the unmodified detectors, so any divergence from $I^\star$ is attributable to the
transport. This is a controlled causal experiment across the seven property classes, not
a statistically representative benchmark.
Each missed incident is labelled by the algebra of Sect.~\ref{sec:guarantees} from the
state a consumer observes: a miss downgraded to \textsf{incomplete} surfaces as an honest
\emph{unknown}, one it cannot flag as a \emph{false all-clear}.

Table~\ref{tab:oracle} reports the outcome. Every one of the shared log's six misses is a
silent false all-clear: it cannot express the gateway-scoped P3.6/P3.7, has no tick for
the silent-node P3.5, and drops the crash-, reorder-, and overload-hit incidents with no
signal an operator could act on. The full fabric absorbs every fault and so reproduces the oracle by
construction; the informative rows are the ablations, where each
single-mechanism removal bar G4, whose cost is timeliness, reintroduces one class of miss:
direct causal evidence that the mechanism, not incidental engineering, preserves the
verdict. The central result is the last column. Removing G1, G2, or G5 leaves the evidence
signal intact (event ids for the G1/G5 losses, the monotonicity check for the G2
reordering), so the miss is \emph{downgraded} to a flagged \textsf{unknown} rather than
reported as a clean all-clear. G3 is the exception: the tick is the very liveness clock
the completeness status is computed against, so removing it makes the two silence
incidents unflaggable and they revert to false all-clears, the one loss class the algebra
cannot surface. G4 preserves the incident \emph{set} but not delivery \emph{latency}, quantified above.

\begin{table}[t]
\centering
\caption{Q8: fault-injected verdict preservation against the fault-free oracle
($|I^\star|{=}7$ incidents spanning P3.1--P3.7). One combined fault campaign (crash,
reorder, silence, overload) is replayed at the alert level under each configuration; incident sets are
computed by the \emph{unmodified} MonPoly engine plus tick-driven reference detectors.
A \emph{false all-clear} is a missed incident the evidence algebra cannot
flag: a silent \textsf{no\_violation}; a \emph{downgraded} miss is instead surfaced
as \textsf{incomplete}.}
\label{tab:oracle}
\footnotesize
\setlength{\tabcolsep}{5pt}
\begin{tabular}{@{}R{2.9cm}ccR{1.7cm}R{2.7cm}@{}}
\toprule
\textbf{Configuration} & \textbf{Preserved} & \textbf{Missed} & \textbf{False all-clears} & \textbf{Downgraded to unknown} \\
\midrule
\rowcolor{rowtint}
\multicolumn{5}{@{}l}{\emph{Baselines}} \\
Shared log $+$ faults & $\mathbf{1/7}$ & $6$ & $\mathbf{6}$ & $0$ \\
\sys (full)           & $\mathbf{7/7}$ & $0$ & $\mathbf{0}$ & $0$ \\
\cmidrule{1-5}
\rowcolor{rowtint}
\multicolumn{5}{@{}l}{\emph{Single-mechanism ablations ($-$G$_i$)}} \\
$-$G1 Delivery & $6/7$ & $1$ & $0$ & $1$ (eid gap) \\
$-$G2 Order      & $6/7$ & $1$ & $0$ & $1$ (order viol.) \\
$-$G3 Tick       & $5/7$ & $2$ & $2$ & $0$ \\
$-$G4 Isolation  & $7/7$ & $0$ & $0$ & $0$ (timeliness) \\
$-$G5 Flow       & $6/7$ & $1$ & $0$ & $1$ (eid gap) \\
\bottomrule
\end{tabular}
\end{table}

\paragraph{Independent re-execution (reproducibility).}
Re-running the public implementation on a \emph{fresh single host} with native brokers
reproduces the qualitative claims with $95\%$ CIs: $0\%$ loss at $250/500/1000$ devices
($1997\pm20$\,msg/s, $p95$ $121\pm6$\,ms at $1000$), saturation near $7$k\,msg/s, isolation
$\sim$208$\times$, and the crash result \emph{exactly}. Absolute figures are host-dependent; the re-execution establishes
reproducibility of the \emph{behaviour}, not of the numbers.

\paragraph{Capacity bounds for scale-out.}
The testbed is lab-scale, so we state the measured figures as \emph{bounds a larger
deployment must respect} rather than as a scale claim. Per durable stream, ingest is
bounded by the publish ceiling ($\approx$10.7\,K\,msg/s single producer, rising to a
$\approx$14\,K\,msg/s plateau at zero loss); past the drain rate the offered load queues
rather than drops, so exceeding the bound costs latency ($p95$ $\sim$1\,s) before it costs
evidence. A fleet is sized against that ceiling: the $1000$-device workload
($\approx$1930\,msg/s) uses under a fifth of it, and read capacity scales independently
with additional L3 consumers. Tolerance to a backend outage is
$T_{\text{buffer}}=C/R_{\text{in}}$, and resuming after a gateway crash costs
$\approx$0.095\,ms per un-acked entry ($474$\,ms for a $5000$-entry replay), so both buffer
capacity and recovery time follow from the fleet rate and the backlog a deployment
permits. These bounds are per-stream and per-gateway and therefore compose across
gateways; what a single host cannot establish is behaviour under a real cross-region WAN
and under correlated multi-gateway partition, which is why the ordering guarantee is
validated across two independent machines above (Q4) and why the two remaining threats
are named explicitly in Sect.~\ref{sec:limits}.

\section{Real-Data Validation: Water-SCADA and IoT/IIoT Datasets}
\label{sec:casestudy}
We validate the fabric on real critical-infrastructure data, replaying two openly published
datasets, WUSTL-IIoT-2021~\cite{wustl2021} (water-supply SCADA) and
TON\_IoT~\cite{moustafa2021toniot,alsaedi2020toniot} (IoT/IIoT), through the real MonPoly
L3 engine. \emph{Scope.} The question is transport, not detection
accuracy: whether real security-event streams can be carried and correlated end-to-end
while preserving the RV assumptions of completeness, ordering, and liveness. These are
flow logs, not distributed deployments, so they exercise detection through the unmodified
engine, not the mechanisms G1--G5; thresholds are fixed a~priori and untuned, and as
neither dataset records a gateway assignment the cross-gateway property runs over a
\emph{synthetic} partition (Sect.~\ref{sec:limits}), each identity hashed to a gateway independent of
label and attack class, so it cannot have been tuned for cross-gateway firings.

Because the device-payload predicates of Sect.~\ref{sec:eval} target application misuse
rather than network-flow malware, we add three flow monitors over \emph{failed}
connections (Zeek \texttt{conn\_state}${=}$\texttt{S0} or Argus \texttt{DstPkts}${=}0$),
applied unchanged to both datasets: \textbf{port\_scan} ($\geq20$ failed destinations in
$60$\,s), \textbf{ddos\_flood} ($\geq100$ connections to one destination in $10$\,s), and
\textbf{c2\_beacon} ($\geq50$ periodic retries, CV${<}0.5$). These feed L3 as
predicate-parametric schemas: $P3.2(A)$ fires when $\geq3$ source identities satisfy alert
$A$ in the window, $P3.6(A)$ when $\geq2$ gateways do; detection is scored per identity.

Both datasets carry and correlate end-to-end. WUSTL-IIoT-2021 is a water-SCADA Modbus
trace ($\approx$1\,M Argus flows, $6$ of $14$ source identities attacking); TON\_IoT a
UNSW cyber-range trace ($461$k records, $19$ of $11{,}536$ identities attacking over nine
classes, stressing correlation over many identities). Per-identity detection
(precision/recall, untuned) is $0.83/0.83$ (WUSTL) and $0.64/0.47$ (TON\_IoT); dedicated
detectors (e.g.\ BATADAL~\cite{taormina2018battle}) score higher. Driving the \emph{real}
MonPoly engine, the coordinated \textbf{P3.2} schema fires $17$/$11$ episodes
(WUSTL/TON\_IoT) and the cross-gateway \textbf{P3.6} is satisfied continuously on WUSTL and
in $168$ episodes on TON\_IoT, at up to four coordinated sources across two gateways,
validating the fleet properties on labelled multi-source traffic.
Application-layer classes (injection, XSS, ransomware, MITM) need content-level
specifications and are out of scope. These counts are \emph{episodes} (consecutive
satisfaction timepoints collapsed), so they are distinct coordinated events rather than
per-timepoint alerts. P3.6's count depends on how the synthetic
partition places identities across gateways and should be read as partition-specific;
P3.2, quantifying over devices, does not.

\section{Related Work}
\label{sec:related}
\sys builds directly on two recent hierarchical RV designs for edge-IoT security: the
first introduces the three-layer hierarchy and its hybrid MonPoly/RTLola backend for
fleet-wide temporal correlation~\cite{ares2026}, the second the formal-methods
specification of the layered security properties~\cite{csr2026}. Both are validated as
single-host prototypes reading a shared log, leaving delivery, and thus resilience, out
of scope. This paper supplies that layer, carrying the same unmodified hierarchy across
the edge-cloud continuum and making evidence completeness part of each verdict.

Measured against the five guarantees, the shared-log RV
prototype~\cite{ares2026,csr2026} and MQTT-only monitoring lack durable delivery, flow
control, independent consumers, a silence-aware clock and trusted ordering, while generic stream
processors supply durability and isolation but no RV-specific ordering and no
critical-infrastructure threat model (Sect.~\ref{sec:threatmodel}).

\paragraph{Hierarchical and edge-based RV.}
Closest in setting is Tsigkanos et al.~\cite{tsigkanos2022edge}, who place MFOTL monitors
on edge nodes in an L0/L1 hierarchy and propagate verdicts between them, adapting MonPoly
for online use. We agree on where monitors belong and differ on what may be assumed
beneath them: their monitors exchange verdicts as REST payloads, and handling lost events
is named as something that ``can be integrated'' rather than provided. \sys supplies that
layer and makes the resulting completeness part of the verdict.
STREL~\cite{strel2017} and SaSTL~\cite{sastl2020} formalise hierarchical and spatially
distributed monitoring, and the hierarchy \sys carries reuses that
pattern~\cite{ares2026,csr2026}; that line develops the \emph{logic} of multi-tier
monitoring, where \sys addresses how such a hierarchy is \emph{deployed} across a real
network.

\paragraph{Distributed, decentralised, and fault-tolerant RV.}
RV at scale has been studied at the algorithm level, in stream-based specification
languages, and in distributed and decentralised
monitoring~\cite{rvsurvey,bauer2012decentralised}, and applied to CPS co-simulation and
digital-twin pipelines~\cite{temperekidis2022fmirv,temperekidis2022dt} and to the
tool-call traces of {LLM} agents~\cite{agentrv2026}. Closest is the
fault-tolerant line: Fraigniaud et al.\ characterise the opinions a fault-tolerant
distributed monitor needs~\cite{fraigniaud2014opinions}, Bonakdarpour et al.\ give
crash-resilient decentralised algorithms~\cite{bonakdarpour2016decentralized}, and Basin
et al.\ monitor MFOTL over \emph{out-of-order} streams~\cite{basin2017outoforder}. That
work makes the monitoring \emph{logic} fault- or reordering-tolerant while abstracting
transport; \sys sits below it, making the delivery layer's loss, ordering and liveness
properties first-class and measurable, so the stream those algorithms consume is the
stream their correctness assumes.

\paragraph{Cloud-native messaging and critical-infrastructure security.}
MQTT~\cite{mqtt,mosquitto}, NATS~\cite{nats} and Kafka~\cite{kafka} are the standard
cloud-native telemetry transports; \sys repurposes that layer for \emph{verification},
requiring the five guarantees of Sect.~\ref{sec:guarantees} rather than best-effort
delivery. Orthogonally, IoT security monitoring spans surveys~\cite{cavalli2024iotsec}, ML
intrusion detection~\cite{santhoshkumar2023mliot}, edge-oriented inference on embedded
platforms~\cite{chintri2026edge}, and attack-path analysis across
critical infrastructures~\cite{stellios2018iot} over the edge-cloud
continuum~\cite{shi2016edge}. The tiered pattern recurs across sectors, including resilience
architectures for electrical power and energy systems~\cite{resilagent2026},
constrained-action response in security operations~\cite{guardrails2026} and
cross-border compliance and forensics platforms~\cite{cyberbridge2026}, and the
same fabric carries a mission-level assurance workload for LLM-assisted ISR
swarms~\cite{isrswarm2026}, so the delivery-layer guarantees are
sector-independent; \sys contributes that delivery layer, not the detection logic.

\section{Limitations and Threats to Validity}
\label{sec:limits}
\emph{Testbed scope.} The throughput, crash-recovery and detection experiments run on a
single containerised host. That establishes that each mechanism holds its stated invariant
under injected faults and that removing it reintroduces one specific failure; it does not
establish absolute throughput or latency, which are host- and network-dependent (the
fresh-host re-execution reproduces the behaviour, not the numbers). Capacity is reported as
bounds a scale-out must respect (Sect.~\ref{sec:eval}), and crash tests cover one gateway's
relay path, not correlated failure. The evaluation is not uniformly single-host: P3.6's
skew resilience is measured across \emph{two independent machines} with separate kernels
and a real network (Q4) against an $\epsilon{\approx}0$ control. Next are multi-host
operation at fleet scale, a cross-region WAN and correlated multi-gateway partition.
\emph{Clock-tick source.} The tick is backend-resident, so silence-liveness survives the
loss of every gateway (verified with both killed); the residual dependency is the backend,
which production should replicate. \emph{Trusted gateway boundary.} The sidecar and outbox
lie inside the trusted computing base, so a compromised gateway that emits plausible
heartbeats while suppressing evidence is out of scope: \sys detects a gateway that goes
\emph{silent} (P3.7) but not one that lies; hash-linked event-sequence commitments would
make suppression tamper-evident. \emph{Real-data study.} It validates that the fleet
properties fire end-to-end on real traces, not the continuity mechanisms, over a synthetic
partition with untuned detection. The crash counts (Q6) and the zero gateway-order
violations (Q1) reproduce exactly, while the device-order count moves with the randomised
time-spoof workload.

\section{Conclusion}
\label{sec:conclusion}
\sys turns a single-host, shared-log RV prototype into a resilient delivery layer for
security monitoring of critical edge-IoT deployments: five broker-contingent continuity
mechanisms add crash- and overload-resilient delivery and close the observed ordering and
silence-liveness gaps, at an overhead small relative to second-scale detection. \sys also
makes evidence completeness part of runtime-verification semantics, so faults and overload
can no longer silently turn missing evidence into a false all-clear: against a fault-free
oracle through the real MonPoly engine, the full fabric preserves all seven injected
incidents whereas the shared log preserves one and misreports the other six. That link
between transport continuity and verdict trustworthiness is the resilience layer missing
between hierarchical RV logic and its deployment on critical infrastructure.

\begin{credits}
\subsubsection{\ackname}
This work has received funding from the European Union's Digital Europe Programme
(CYBERGUARD, No~101190251) and from its Horizon Europe research and innovation
programme through the Smart Networks and Services Joint Undertaking (XTRUST-6G,
No~101192749).

\subsubsection{\discintname}
The authors have no competing interests to declare that are relevant to the content of
this article.
\end{credits}

\bibliographystyle{splncs04}
\bibliography{references}

\end{document}